\documentclass[11pt]{article}
\usepackage{graphicx, amsmath,amsthm,amssymb,url}

\usepackage{mathtools}
\usepackage[multidot]{grffile}
\usepackage[normalem]{ulem}
\usepackage{cite}
\usepackage{amsmath,amssymb,amsfonts}
\usepackage{graphicx}
\usepackage{textcomp}
\usepackage{algorithm}
\usepackage{algorithmicx}
\usepackage{algpseudocode}
\usepackage{amssymb}

\usepackage[dvipsnames]{xcolor}

\usepackage{comment}

\usepackage{fullpage,etoolbox}
\usepackage{color}
\usepackage{hyperref}
\hypersetup{
    colorlinks=true,
    linkcolor=blue,
    citecolor = blue,
    urlcolor = blue
}

\newtheorem{theorem}{Theorem}

\newtheorem{definition}{Definition}
\newtheorem{lemma}[theorem]{Lemma}

\usepackage[algo2e,ruled]{algorithm2e}
\usepackage{makecell,booktabs,ctable}

\begin{document}

\title{A Scalable Game Theoretic Approach for Coordination of Multiple Dynamic Systems}

\author{Mostafa M. Shibl\footnote{Elmore Family School of Electrical and Computer Engineering, Purdue University, West Lafayette, IN, USA. Emails: \url{mabdelna@purdue.edu}, \url{gupta869@purdue.edu}}\qquad Vijay Gupta$^*$}
\date{}

\maketitle

\begin{abstract}
Learning in games provides a powerful framework to design control policies for self-interested agents that may be coupled through their dynamics, costs, or constraints. We consider the case where the dynamics of the coupled system can be modeled as a Markov potential game. In this case, distributed learning by the agents ensures that their control policies converge to a Nash equilibrium of this game. However, typical learning algorithms such as natural policy gradient require knowledge of the entire global state and actions of all the other agents, and may not be scalable as the number of agents grows. We show that by limiting the information flow to a local neighborhood of agents in the natural policy gradient algorithm, we can converge to a neighborhood of optimal policies. If the game can be designed through decomposing a global cost function of interest to a designer into local costs for the agents such that their policies at equilibrium optimize the global cost, this approach can be of interest to team coordination problems as well. We illustrate our approach through a sensor coverage problem.
\end{abstract}

\section{Introduction}
Many problems can be represented through the interaction  of multiple dynamic systems %is a problem of fundamental interest in control theory. To ensure scalability as the number of decision makers increases, distributed and decentralized control is a must. Much work has been done in these fields for a {\em team} formulation in which all the decision makers are interested in optimizing a system cost function. We consider the situation 
when each system (also called decision maker or agent) is self-interestedly optimizing its own individual cost function. If the dynamics, costs, or constraints of the agents are coupled, they still need to consider the effect of each other's decisions while planning their policies. Game theory has a rich history of being used in such contexts. Although initial results tended to analyze the equilibrium policies of the agents, more recently, the use of learning algorithms in games has been investigated as a tool for the agents to learn their individual policies. 

Specifically, utilizing learning algorithms to coordinate the decisions of self-interested agents in one-shot (meaning that agents must make one decision each) optimization problems using repeated static games has been explored in the literature~\cite{jm2,jm3,jm1}. In repeated static games, the utilities of the agents are not coupled in time. When stage games are potential games, learning algorithms display nice convergence properties to the equilibria of the stage game~\cite{nl}. If the local cost functions at the agents are defined such that the agents' policies at the equilibrium optimize a global cost function of interest to a designer, this formalism can be used to solve static team coordination problems, as has been explored for problems such as resource allocation or sensor coverage~\cite{jm2,jm3,jm1}.

For multi-stage optimization problems, dynamic games need to be considered. We focus on Markov games in which the agents interact with the environment and their actions have an effect on the environment (through a system state that evolves stochastically) as well as the utilities of other agents. These games require the agents to optimize utility functions that are summed over a time horizon. %Both equilibrium analysis and learning algorithms in such problems are more difficult since the decisions of the agents impact not only their immediate payoffs but also the state transitions of the underlying dynamic system. 
The design of learning algorithms in Markov games can, in general, be cast as a multi-agent reinforcement learning problem~\cite{lb,kz}. However, convergence analysis can become difficult for general Markov games. Various learning algorithms have been considered and analyzed for special cases, e.g., decentralized Q-learning in infinite-horizon discounted zero-sum Markov games~\cite{ms2}, fictitious play for zero-sum and identical-interest Markov games~\cite{ms1}, and multi-player zero-sum Markov games with networked separable interactions~\cite{cp}. In this paper, we will focus on independent natural policy gradient learning. This algorithm is a generalization of the natural policy gradient for policy optimization in Markov Decision Processes (MDPs)~\cite{sk2} and inherits similar distributability, efficiency, and convergence guarantees in Markov games~\cite{cd,aoms,rf,dd1,cm}. %For zero-sum Markov games, such algorithms have been shown to have global convergence guarantees to Nash equilibria~\cite{cd, aoms}.

Markov Potential Games (MPGs) are a special type of Markov games that are defined in terms of a potential function similar to static potential games. These games display several desirable properties including nice characterizations of Nash equilibria, uniqueness of the equilibrium under some additional conditions, guaranteed convergence of many learning algorithms to the Nash equilibria, and efficiency of the Nash equilibria. Specifically for MPGs, independent policy gradient and independent natural policy gradient were shown to have nice convergence behavior~\cite{sl,rf,ding2022independent}. Other learning algorithms such as modified Q-learning and networked policy gradient have also been shown to converge to suitable equilibria in such games~\cite{cm,sa}. In our paper, we concentrate on this class of games.  %However, for the specific case of Markov potential games (MPGs), previous literature have identified independent learning algorithms with guarantees of convergence to a unique optimal Nash equilibrium. 

However, almost all existing works in convergence of learning algorithms in Markov games, and MPGs in particular, assume %full state and action observability by all the agents. In other words, it is assumed 
that the agents have full information of the state and actions of all other agents at every time step to update their policy. As the number of agents increases, this amount of communication and computation at each agent might not be sustainable. To ensure scalability as the number of agents increases, one approach in game theory utilizes population games and mean-field games, where the information from all the other agents can be summarized into a few `average' variables about the state of the environment or sub-populations of the agents~\cite{nm3,nm1,nm2,pc1,pc2,pc3}. However, such approaches may not be possible when the number of agents is not large enough. %However, these learning algorithms usually require each agent to have access to the entire state, actions, and rewards for all other agents. This is obviously not scalable and is a major reason why such learning algorithms have not been used to coordinate dynamical systems in practice. 

In this paper, we use a different approach. Inspired by a framework for cooperative systems recently developed in~\cite{gq,gq2}, we show that the error incurred in the policy design can be guaranteed to be limited if the data from only the neighboring agents (based on a coupling graph) are used in the learning algorithm. Specifically, we consider independent natural policy gradient for MPGs and establish an exponential decay property. Intuitively, this property means that the effect of an agent on the policies of the other agents decreases exponentially with the distance between those agents according to the coupling graph. This property can be used to develop a scalable framework for distributed decision making for dynamic systems through local observability at the cost of a limited and guaranteed loss of performance, in the sense of convergence to an $\epsilon$-neighborhood of a Nash equilibrium (NE). %Our main contribution is to develop learning algorithms for MPGs in order to design the control policies for coordination of multiple dynamic systems with a focus on scalability by showing that by using the values from only the neighboring agents (as defined according to a coupling graph) in a natural policy gradient algorithm, the policies of the agents can be guaranteed to converge to an $\epsilon$-neighborhood of those at a Nash equilibrium (NE). 
We demonstrate our algorithms through a sensor coverage application. %By incorporating both the strategic behavior of agents and the dynamics of the system, MPGs offer a principled approach to decentralized control, enabling effective coordination in complex, large scale environments. This research paper aims to explore the application of MPGs for the coordination of multiple dynamic systems, with a specific focus on scalability. By leveraging the inherent structure of Markov potential games and employing policy gradient algorithms known for convergence to optimal equilibrium points, we seek to develop a scalable framework that enables decentralized control in large scale environments.

The paper is organized as follows. Section~\ref{sec:model} introduces the model used. Our algorithm is proposed and analyzed in Section~\ref{sec:proposed}. Section~\ref{sec:examples} applies the algorithm to two problems. Section~\ref{sec:conclusions} concludes the paper with some future directions.

%The need for agents to autonomously learn and adapt was highlighted by Busoniu et al. \cite{lb}, where the challenges and techniques of MARL across various applications was reviewed. The review included emphasized different learning goals and techniques, as well as the scalability issues of algorithms. Moreover, it was shown that advances in single-agent RL techniques and deep learning have helped in improving different aspects of MARL, including theoretical and algorithmic aspects of MARL, cooperative and competitive frameworks of MARL, full and partial observability, and the relationship between MARL and game theory \cite{kz, vm, ao}.

%Utility design within game theoretic models for distributed resource allocation was explored by Marden and Adam \cite{jm2}. Different methodologies were proposed that ensure desirable properties such as equilibrium existence and optimality. The authors in \cite{jm3} explored the application of game theoretic methods to cooperative control problems. Game theoretic frameworks were utilized to improve the management of multi-agent systems in both system performance and robustness. Also, this work was extended by Marden \cite{jm1} to state-based potential games. This new framework improved coordination in practical applications such as distributed resource allocation. Further, such systematic methodologies can be utilized for control law design through designing local agent objective functions that align with system level objectives and guarantee Nash equilibria corresponding to the optimal solutions of these objectives \cite{nl}.

\section{Model}
\label{sec:model}
\subsection{Markov Potential Game}

%A Markov potential game (MPG) is a special type of a multi-agent Markov decision process, which is a mathematical construction that is utilized to model sequential decision making in situations where outcomes are partially random and partially controlled by multiple decision makers.

Consider a network of $n$ agents represented as an undirected graph $\mathcal{G} = (\mathcal{N} , \mathcal{E})$, where $\mathcal{N} = {1, 2, ..., n}$ represents the set of agents and $\mathcal{E} \subset \mathcal{N} \times \mathcal{N}$ represents the set of edges, where each edge represents the coupling between two agents through their rewards or dynamics, as explained below. Let $\mathcal{N}_i$ denote the set of neighbors of node $i$ according to $\mathcal{G}$ and include $i$ in $\mathcal{N}_{i}$ by convention. Similarly, let $\mathcal{N}_i^\kappa$ denote the $\kappa$-hop neighborhood of $i$ that contains the agents whose graph distance to $i$ has length less than or equal to $\kappa$.

%Each agent $i$ can take an action $a_i \in \mathcal{A}_i$ and has a state $s_i \in \mathcal{S}_i$. The global action is characterized by $a = (a_1, a_2, ..., a_n) \in \mathcal{A} = \mathcal{A}_1 \times \mathcal{A}_2 \times ... \times \mathcal{A}_n$. Likewise, the global state is characterized by $s = (s_1, s_2, ..., s_n) \in \mathcal{S} = \mathcal{S}_1 \times \mathcal{S}_2 \times ... \times \mathcal{S}_n$. 

Consider an MPG defined through the tuple $(\mathcal{N}, \{\mathcal{S}_i\}_{i \in \mathcal{N}}, \{\mathcal{A}_i\}_{i \in \mathcal{N}}, \{r_i\}_{i \in \mathcal{N}}, P, \gamma, \mu, \Phi^\pi$), where:
\begin{itemize}
  \item $\mathcal{N} = \{1, 2, \ldots, n\}$ is the set of agents.
  \item $\mathcal{S}_i$ represents a finite set of states for agent $i$, with $s_{i}^{t}$ denoting the state of agent $i$ at time $t$. To avoid notational clutter, we sometimes suppress the dependence on time. We write $\mathcal{S} = \prod_{i \in \mathcal{N}} \mathcal{S}_i$, $\mathcal{S}_{\mathcal{N}_{i}} = \prod_{j \in \mathcal{N}_{i}} \mathcal{S}_j$, and $\mathcal{S}_{-i} = \prod_{j \neq i} \mathcal{S}_j$ to denote, respectively, the joint state space of all agents, the joint state space of agents in the neighbor set of agent $i$, and of all agents except agent $i$. Also, we use the notation $s = (s_i, s_{-i}) \in \mathcal{S}$ in which $s_i \in \mathcal{S}_i$ and $s_{-i} \in \mathcal{S}_{-i}$. %When we need to denote the state at time $t$, we use the notation $s^t$.
  \item $\mathcal{A}_i$ represents a finite set of actions for agent $i$, with $a_{i}^{t}$ denoting the action of agent $i$ at time $t$. Once again, when not needed, we sometimes suppress the dependence on time. We define the notations $\mathcal{A}$, $\mathcal{A}_{\mathcal{N}_{i}}$, $\mathcal{A}_{-i}$, and $a$ analogously to their counterparts in the state space. %Using common notations in game theory, we write $\mathcal{A} = \prod_{i \in \mathcal{N}} \mathcal{A}_i$ and $\mathcal{A}_{-i} = \prod_{j \neq i} \mathcal{A}_j$ to denote, respectively, the joint action space of all agents and of all agents except $i$. Also, we use the notation $a = (a_i, a_{-i}) \in \mathcal{A}$ in which $a_i \in \mathcal{A}_i$ and $a_{-i} \in \mathcal{A}_{-i}$.
  \item $r_i: \mathcal{S}_{\mathcal{N}_i} \times \mathcal{A}_{\mathcal{N}_i} \rightarrow [0, r_{max}]$ denotes the individual reward function of agent $i \in \mathcal{N}$. Note that unlike a general MPG, in our formulation, this reward is dependent only on the states and actions of agents in the neighbor set of agent $i$. %Agent $i$ receives a reward $r_i(s_i, a_i, s_{\mathcal{N}_i}, a_{\mathcal{N}_i})$ when it takes action $a_i$ in state $s_i$, and the neighboring agents take actions $a_{\mathcal{N}_i}$ in state $s_{\mathcal{N}_i}$.
  \item $P$ is the transition probability matrix, where $P(s' | s, a)$ represents the probability of moving from state $s$ to $s'$ when agents select a joint action $a \in \mathcal{A}$. Once again, in our formulation, we assume that this probability can be written as %Given current state and action, the next individual state is independently generated and is only dependent on the neighbors of agent $i$, as shown in (\ref{eq:statetrans}).
    \begin{equation}
    \label{eq:statetrans}
    P(s'|s, a) = \prod_{i=1}^n P(s_i'|s_{\mathcal{N}_i}, a_{\mathcal{N}_i}),
    \end{equation}
    so that the next state of agent $i$ is dependent only on the states and actions of agents in the neighbor set of agent $i$ according to graph $\mathcal{G}$.
  \item $\gamma \in (0, 1)$ is the discount factor that is assumed to be uniform for all agents.
  \item $\mu \in \Delta(\mathcal{S})$ is the distribution of the initial state $s^0$.
  \item $\Phi^\pi(s)$ denotes the potential function for the MPG.
\end{itemize}

A deterministic, stationary policy $\pi_i: \mathcal{S} \rightarrow \mathcal{A}_i$ determines the action that agent $i$ will take in each state $s \in \mathcal{S}$, for each agent $i \in \mathcal{N}$. Furthermore, for each agent $i \in \mathcal{N}$, a stochastic policy $\pi_i: \mathcal{S} \rightarrow \Delta(\mathcal{A}_i)$ defines a probability distribution over the actions of agent $i$ for each state $s \in \mathcal{S}$. Let $\Delta(\mathcal{A}_i)^{|\mathcal{S}|}$, $\Delta(\mathcal{A})^{|\mathcal{S}|}$, and $\Delta(\mathcal{A}_{-i})^{|\mathcal{S}|}$ denote the sets representing all stochastic policies for agent $i$, all joint stochastic policies for all agents, and all joint stochastic policies for all agents except agent $i$, respectively. A joint policy $\pi$ induces a distribution over the trajectories $\tau = \{(s^t, a^t, r^t)\}_{t \geq 0}$.

Each agent tries to maximize their value function, shown in (\ref{eq:valuefunc}), which also maximizes their Q-function, shown in (\ref{eq:qfunc}), with the expectation taken over the trajectory $\tau$.

\begin{equation}
\label{eq:valuefunc}
V^\pi_i(s) = \mathbb{E} \Bigg[ \sum_{t=0}^\infty \gamma^t r_i(s^t, a^t) \Big| s^0 = s \Bigg]
\end{equation}

\begin{equation}
\label{eq:qfunc}
Q^\pi_i(s, a) = \mathbb{E} \Bigg[ \sum_{t=0}^\infty \gamma^t r_i(s^t, a^t) \Big| s^0 = s, a^0 = a \Bigg]
\end{equation}

We define the total variation (TV) distance between two probability measures $A$ and $B$ in Definition \ref{def:TVD}.

\begin{definition}
\label{def:TVD}
Consider two probability measures $A$ and $B$ defined on the measurable space $(\Omega, \mathcal{F})$. The total variation distance between $A$ and $B$ is defined below, and takes values in the range $[0, 1]$.
\begin{center}
$TV(A, B)= \sup_{C \in \mathcal{F}} |A(C) - B(C)|$
\end{center}
\end{definition}

For all policies $\pi_i, \pi'_i \in \Delta(\mathcal{A}_i)^{|\mathcal{S}|}$, all states $s \in \mathcal{S}$, and all agents $i \in \mathcal{N}$, the potential function of the MPG satisfies
\begin{equation}
\label{eq:potential}
\Phi^{(\pi_i, \pi_{-i})}(s) - \Phi^{(\pi'_i, \pi_{-i})}(s) = V^{(\pi_i, \pi_{-i})}_i(s) - V^{(\pi'_i, \pi_{-i})}_i(s).
\end{equation}
The maximization of the value functions of all agents leads to maximization of the potential function. The equilibrium joint policy is defined as follows.

\begin{definition}[Equilibrium and $\epsilon$-Equilibrium Joint Policies]
\label{def:optpolicy}
 A joint policy $\pi^* = (\pi^*_i)_{i \in \mathcal{N}} \in \Delta(\mathcal{A})^{|\mathcal{S}|}$ is a (Nash) equilibrium policy if for each agent $i$, for all policies $\pi_i \in \Delta(\mathcal{A}_i)^{|\mathcal{S}|}$, and all $s \in \mathcal{S}$,
\begin{equation}
V^{(\pi^*_i, \pi^*_{-i})}_i(s) \geq V^{(\pi_i, \pi^*_{-i})}_i(s). 
\end{equation}
\\ \\ Further, given an $\epsilon > 0$, a joint policy $\pi^* = (\pi^*_i)_{i \in \mathcal{N}} \in \Delta(\mathcal{A})^{|\mathcal{S}|}$ is an $\epsilon$-equilibrium policy if for each agent $i$, for all policies $\pi_i \in \Delta(\mathcal{A}_i)^{|\mathcal{S}|}$, and all $s \in \mathcal{S}$,
\begin{equation}
V^{(\pi^*_i, \pi^*_{-i})}_i(s) \geq V^{(\pi_i, \pi^*_{-i})}_i(s) - \epsilon.
\end{equation}

\end{definition}

%In MPGs, the policies that maximize the potential function form an Nash equilibrium (NE) policy, that is sometimes called the optimal Nash equilibrium of the game. 

%In this paper, we aim to design control policies for self-interested agents (i.e. maximize their own value functions (\ref{eq:valuefunc})) in a multi-agent system that are interconnected through dynamics and costs. The aim of the optimal joint policy, $\pi^*$, is to maximize the objective function of the system, which is the potential function in an MPG. In an MPG, 

\subsection{Independent Natural Policy Gradient Algorithm}
We assume that independent natural policy gradient algorithm is being used by the agents, which is known to converge to the optimal NE in MPGs where the policies that maximize the potential function. First, we assume that the policy of the agent $i$ uses the soft-max parameterization as  
\begin{equation}
\label{eq:policy}
\pi_{\theta_i}(a_i|s) = \frac{exp(\theta_{i,s,a_i})}{\sum_{a_i' \in \mathcal{A}_i} exp(\theta_{i,s,a_i'})} \text{ for all } s \in \mathcal{S} \text{ and } a_i \in \mathcal{A}_i.
\end{equation}
Denote the stacked version of $\theta_{i,s,a_i'}$ for all states and actions by $\theta_i\in\mathbb{R}^{|\mathcal{A}_i| \times |\mathcal{S}|}$. Then, the independent natural policy gradient algorithm updates this parameter vector as 
\begin{equation}
\label{eq:inpg1}
\theta_i^{(t+1)} = \theta_i^{(t)} + \eta F^{-1}_{i,\mu}(\theta^{(t)}) \nabla_{\theta_i} V_i^{\pi_{\theta_i}(t)},
\end{equation}
where $\eta$ denotes the step size and $F^{-1}_{i,\mu}(\theta)$ denotes the pseudoinverse of the Fischer information matrix defined as 
\begin{equation*}
%\label{eq:fischer}
F_{i,\mu}(\theta) = \mathbb{E}_{s \sim d_\mu^{\pi_{\theta_i}}} \mathbb{E}_{a_i \sim \pi_{\theta_i}(\cdot|s)} \Big [ \nabla_{\theta_i} \log \pi_{\theta_i}(a|s) \log \pi_{\theta_i}(a|s)^\top \Big ].
\end{equation*}
The update rule in (\ref{eq:inpg1}) can be streamlined as 
\begin{equation}
\label{eq:inpg2}
\theta_{i,s,a_i}^{t+1} = \theta_{i,s,a_i}^{t} + \frac{\eta}{1-\gamma} A_i^{\pi_{\theta_i}}(s,a),
\end{equation}
using the advantage function for each agent defined as
\begin{equation}
\label{eq:advfunc}
A^\pi_i(s, a) = Q^\pi_i(s, a) - V^\pi_i(s).
\end{equation}
The following result is standard.

\begin{theorem}[Theorem 1.1 in \cite{rf}]
Consider an MPG in which all agents update their policies according to independent natural policy gradient algorithm. For a sufficiently small step size $\eta$, independent natural policy gradient exhibits last-iterate (asymptotic) convergence to the optimal Nash equilibrium policy.
\end{theorem}

\subsection{Problem Considered}
The independent natural policy gradient algorithm updates need every agent to have access to the global state $s$ and global action $a$. To alleviate the communication and computational burden, we are interested in the following problem. Suppose that the states and actions of only the agents within its $\kappa$-hop neighborhood (denoted by $\mathcal{N}_{i}^{\kappa}$) according to the graph $\mathcal{G}$  are available to agent $i$ at every time. We wish to characterize if the independent natural policy gradient algorithm still converges, and if so, whether it converges to an $\epsilon$-NE where $\epsilon$ is a function of $\kappa$.

\section{Proposed Algorithm}
\label{sec:proposed}
In order to utilize state and action information from neighbors only, the independent natural policy gradient has to be modified to be dependent on the available limited information only. Analogous to the notation defined earlier, we stack the states and actions of all agents in $\mathcal{N}_{i}^{\kappa}$ in the vectors $s_{\mathcal{N}_i^\kappa} \in \mathcal{S}_{\mathcal{N}_i^\kappa} $ and $a_{\mathcal{N}_i^\kappa}\in \mathcal{A}_{\mathcal{N}_i^\kappa} $, respectively. Similarly, we consider the global state and the global action to be decomposed into two parts as $s=(s_{\mathcal{N}_i^\kappa},s_{-\mathcal{N}_i^\kappa})$ and $a=(a_{\mathcal{N}_i^\kappa},a_{-\mathcal{N}_i^\kappa})$, respectively. For all $s_{\mathcal{N}_i^\kappa}$ and actions $a_i$ of the $i$-th agents, we consider that the policy of the agent $i$ is represented as

\begin{equation}
\label{eq:mpolicy}
\pi_{\theta_{i,\kappa}}(a_i|s_{\mathcal{N}_i^\kappa}) = \frac{exp(\theta_{i,s_{\mathcal{N}_i^\kappa},a_i)}}{\sum_{a_i' \in \mathcal{A}_i} exp(\theta_{i,s_{\mathcal{N}_i^\kappa},a_i'})}.
\end{equation}

The modified update rule for the independent natural policy algorithm for each agent $i$ is now described as
\begin{equation}
\label{eq:minpg}
\theta_{i,s_{\mathcal{N}_i^\kappa},a_i}^{t+1} = \theta_{i,s_{\mathcal{N}_i^\kappa},a_i}^{t} + \frac{\eta}{1-\gamma} A_i^{\pi_{\theta_{i,\kappa}}}(s_{\mathcal{N}_i^\kappa},a_{\mathcal{N}_i^\kappa}).
\end{equation}

The main result of the paper is presented in Theorem \ref{theorem:main}.
\begin{theorem}
\label{theorem:main} Consider the problem formulated in Section~\ref{sec:model}. The modified independent natural policy gradient algorithm~(\ref{eq:minpg}) converges to an $\epsilon$-equilibrium joint policy with $\epsilon = \frac{r_{max}}{1-\gamma} \gamma^{\kappa+1}$.
\end{theorem}

Intuitively, the proof of Theorem \ref{theorem:main} is based on an exponential decay property of the Q-function in its dependence on the state and action information from other agents. %This means the influence of distant agents from agent $i$ on the Q-function of agent $i$ decreases exponentially with their graph distance $\kappa$ at a rate of $\gamma$, which allows the Q-function of agent $i$ to be approximated using information from neighboring agents only, rather than the complete state and action information from all agents.
%
%In order to prove Theorem \ref{theorem:main} mathematically, we introduce a set of assumptions and lemmas. In order to ensure the non-negativity of rewards, we utilize Assumption \ref{assump:posrew} in the below formulation.
%
% \begin{assumption}
% \label{assump:posrew}
% We assume $0 \leq r_i(s_{\mathcal{N}_i^\kappa}, a_{\mathcal{N}_i^\kappa}) \leq r_{max}$ for all agents $i \in \mathcal{N}$, all states $s{\mathcal{N}_i^\kappa} \in \mathcal{S}{\mathcal{N}_i^\kappa}$, and all actions $a_{\mathcal{N}_i^\kappa} \in \mathcal{A}{\mathcal{N}_i^\kappa}$.
% \end{assumption}
To prove this result, we begin with the following intermediate results. Lemma \ref{lemma:deltaq} provides an upper bound on the change in the value of the Q-function when a change in the global state and joint action occurs.

\begin{lemma}
\label{lemma:deltaq}
Consider the problem formulation in Section~\ref{sec:model}, and the modified independent natural policy gradient algorithm~(\ref{eq:minpg}). For any agent $i$, consider states $s,s'\in\mathcal{S}$ such that $s = (s_{\mathcal{N}_i^{\kappa^2}}, s_{\mathcal{N}_{-i}^{\kappa^2}})$ and $s' = (s_{\mathcal{N}_i^{\kappa^2}}, s'_{\mathcal{N}_{-i}^{\kappa^2}})$. Consider a joint policy $\pi$ parameterized as~(\ref{eq:mpolicy}). Similarly, consider actions $a,a'\in\mathcal{A}$ related through $a = (a_{\mathcal{N}_i^{\kappa^2}}, a_{\mathcal{N}_{-i}^{\kappa^2}})$, and $a' = (a_{\mathcal{N}_i^{\kappa^2}}, a'_{\mathcal{N}_{-i}^{\kappa^2}})$. For any joint policy $\pi$ followed by the agents, the relation $|Q_i^{\pi}(s, a) - Q_i^{\pi}(s', a')| \leq k\gamma^{\kappa+1}$ holds.
\end{lemma}

\begin{proof}
Let $\rho_i^t$ be the distribution of $(s_i^t, a_i^t)$, conditioned on $(s^0, a^0) = (s, a)$, and $\rho\textquotesingle_i^t$ be the distribution of $(s_i^t, a_i^t)$, conditioned on $(s^0, a^0) = (s', a')$. $\rho_i^t$ is dependent on $(s_{\mathcal{N}_i^t}, a_{\mathcal{N}_i^t})$, while $\rho\textquotesingle_i^t$ is dependent on $(s'_{\mathcal{N}_i^t}, a'_{\mathcal{N}_i^t})$. However, for $t \leq \kappa$, since the initial states and actions are chosen such that the states and actions of the $\kappa^2$-hop neighborhood of agent $i$ are unchanged and the state and action evolution depends on the $\kappa^2$-hop neighborhood of agent $i$ only up to time $t = \kappa$, and the localized dependence structure in (\ref{eq:statetrans}), we have $(s_{\mathcal{N}_i^t}, a_{\mathcal{N}_i^t})$ is the same as $(s'_{\mathcal{N}_i^t}, a'_{\mathcal{N}_i^t})$. Thus, we conclude that $\rho_i^t = \rho\textquotesingle_i^t$ for $t \leq \kappa$, because the state-action evolution for agent $i$ up to time $t = \kappa$ is the same.

\begin{equation*}
\begin{split}
|Q_i^{\pi}(s, a) - Q_i^{\pi}(s', a')| &\stackrel{\text{def}}{=} \Bigg| \sum_{t=0}^\infty \mathbb{E}[\gamma^t r_i(s_i^t, a_i^t) | s^0 = s, a^0 = a] - \sum_{t=0}^\infty \mathbb{E}[\gamma^t r_i(s_i^t, a_i^t)| s^0 = s', a^0 = a'] \Bigg| \\
&= \Bigg| \sum_{t=0}^\infty \mathbb{E}[\gamma^t r_i(s_i^t, a_i^t) | s^0 = s, a^0 = a] - \mathbb{E}[\gamma^t r_i(s_i^t, a_i^t)| s^0 = s', a^0 = a'] \Bigg| \\
\stackrel{\text{triangle inequality}}{\leq} &\sum_{t=0}^\infty \Big| \mathbb{E}[\gamma^t r_i(s_i^t, a_i^t) | s^0 = s, a^0 = a] - \mathbb{E}[\gamma^t r_i(s_i^t, a_i^t)| s^0 = s', a^0 = a'] \Big| \\
\stackrel{\rho_i^t = \rho\textquotesingle_i^t \text{ for } t \leq \kappa}{=} &\sum_{t=\kappa+1}^\infty \Big| \gamma^t \mathbb{E}_{(s_i, a_i) \sim \rho_i^t} r_i(s_i, a_i) - \gamma^t \mathbb{E}_{(s_i, a_i) \sim \rho\textquotesingle_i^t} r_i(s_i, a_i) \Big| \\
\stackrel{\text{due to Definition \ref{def:TVD} of TV distance}}{\leq} & \sum_{t=\kappa+1}^\infty \gamma^t r_{max} TV(\rho_i^t, \rho\textquotesingle_i^t) \\
&\leq \frac{r_{max}}{1-\gamma} \gamma^{\kappa+1} \\
&= k\gamma^{\kappa+1}
\end{split}
\end{equation*}
\end{proof}

Based on Lemma \ref{lemma:deltaq}, we can provide an upper bound for the difference between the approximated Q-function, calculated from the truncated state and joint action information, and the true Q-function. Similar to the work in \cite{gq}, a class of truncated Q-functions defined in (\ref{eq:trunc1} - \ref{eq:trunc2}) is utilized, where $\kappa' = \kappa^2$, and $w_i(s_{\mathcal{N}_{-i}^{\kappa'}}, a_{\mathcal{N}_{-i}^{\kappa'}}; s_{\mathcal{N}_{i}^{\kappa'}}, a_{\mathcal{N}_{i}^{\kappa'}})$ are non-negative scalar values. 

\begin{equation}
\begin{split}
\label{eq:trunc1}
\hat{Q}_i^{\pi}(s_{\mathcal{N}_i^{\kappa'}}, a_{\mathcal{N}_i^{\kappa'}}) = \sum_{s_{\mathcal{N}_{-i}^{\kappa'}}, a_{\mathcal{N}_{-i}^{\kappa'}}} w_i(&s_{\mathcal{N}_{-i}^{\kappa'}}, a_{\mathcal{N}_{-i}^{\kappa'}}; s_{\mathcal{N}_{i}^{\kappa'}}, a_{\mathcal{N}_{i}^{\kappa'}}) Q_i^{\pi}(s_{\mathcal{N}_{i}^{\kappa'}}, s_{\mathcal{N}_{-i}^{\kappa'}}, a_{\mathcal{N}_{i}^{\kappa'}}, a_{\mathcal{N}_{-i}^{\kappa'}})
\end{split}
\end{equation}

\begin{equation}
\label{eq:trunc2}
\sum_{s_{\mathcal{N}_{-i}^{\kappa'}}, a_{\mathcal{N}_{-i}^{\kappa'}}} w_i(s_{\mathcal{N}_{-i}^{\kappa'}}, a_{\mathcal{N}_{-i}^{\kappa'}}; s_{\mathcal{N}_{i}^{\kappa'}}, a_{\mathcal{N}_{i}^{\kappa'}}) = 1
\end{equation}

\begin{lemma}
\label{lemma:truncatedQ}
Consider the problem formulation in Section~\ref{sec:model} and the modified independent natural policy gradient algorithm~(\ref{eq:minpg}). For any agent $i$, consider the state $s_{\mathcal{N}_{-i}^{\kappa'}}\in\mathcal{S}_{\mathcal{N}_{-i}^{\kappa'}}$ and action $a_{\mathcal{N}_{-i}^{\kappa'}}\in\mathcal{A}_{\mathcal{N}_{-i}^{\kappa'}}$ for the truncated Q-function $\hat{Q}_i^{\pi}$ in~(\ref{eq:trunc1}), and state $s\in\mathcal{S}$ and action $a\in\mathcal{A}$ for the true Q-function $Q_i^{\pi}$. For any joint policy $\pi$ parameterized as~(\ref{eq:mpolicy}) followed by the agents, the relation $|\hat{Q}_i^{\pi}(s_{\mathcal{N}_i^{\kappa'}}, a_{\mathcal{N}_i^{\kappa'}}) - Q_i^{\pi}(s, a)| \leq k\gamma^{\kappa+1}$ holds.
\end{lemma}

\begin{proof}
We begin the proof using the definition of the truncated and true Q-functions.

\begin{equation*}
\begin{split}
|\hat{Q}_i^{\pi}(s_{\mathcal{N}_i^{\kappa'}}, a_{\mathcal{N}_i^{\kappa'}}) - Q_i^{\pi}(s, a)| &\stackrel{\text{def}}{=} \Big| \sum_{s'_{\mathcal{N}_{-i}^{\kappa'}}, a'_{\mathcal{N}_{-i}^{\kappa'}}} w_i(s'_{\mathcal{N}_{-i}^{\kappa'}}, a'_{\mathcal{N}_{-i}^{\kappa'}}; s_{\mathcal{N}_{i}^{\kappa'}}, a_{\mathcal{N}_{i}^{\kappa'}}) Q_i^{\pi}(s_{\mathcal{N}_{i}^{\kappa'}}, s'_{\mathcal{N}_{-i}^{\kappa'}}, a_{\mathcal{N}_{i}^{\kappa'}}, a'_{\mathcal{N}_{-i}^{\kappa'}}) - \\ 
&~~~~~~~~~~ Q_i^{\pi}(s_{\mathcal{N}_{i}^{\kappa'}}, s_{\mathcal{N}_{-i}^{\kappa'}}, a_{\mathcal{N}_{i}^{\kappa'}}, a_{\mathcal{N}_{-i}^{\kappa'}}) \Big| \\
&\leq \sum_{s'_{\mathcal{N}_{-i}^{\kappa'}}, a'_{\mathcal{N}_{-i}^\kappa}} w_i(s'_{\mathcal{N}_{-i}^{\kappa'}}, a'_{\mathcal{N}_{-i}^{\kappa'}}; s_{\mathcal{N}_{i}^{\kappa'}}, a_{\mathcal{N}_{i}^{\kappa'}}) \Big| Q_i^{\pi}(s_{\mathcal{N}_{i}^{\kappa'}}, s'_{\mathcal{N}_{-i}^{\kappa'}}, a_{\mathcal{N}_{i}^{\kappa'}}, a'_{\mathcal{N}_{-i}^{\kappa'}}) - \\ 
&~~~~~~~~~~ Q_i^{\pi}(s_{\mathcal{N}_{i}^{\kappa'}}, s_{\mathcal{N}_{-i}^{\kappa'}}, a_{\mathcal{N}_{i}^{\kappa'}}, a_{\mathcal{N}_{-i}^{\kappa'}}) \Big| \\
&\stackrel{\text{from (\ref{eq:trunc2}) and Lemma \ref{lemma:deltaq}}}{\leq} k\gamma^{\kappa+1}
\end{split}
\end{equation*}
\end{proof}

Furthermore, we are able to calculate an upper bound for the difference between the gradient step in independent natural policy gradient in (\ref{eq:inpg2}) and the gradient step in modified independent natural policy gradient in (\ref{eq:minpg}).

\begin{lemma}
\label{lemma:gradstep}
Consider the problem formulation in Section~\ref{sec:model} and the modified independent natural policy gradient algorithm~(\ref{eq:minpg}). The relationship between the true and the truncated gradient step $|\hat{\nabla}_{\theta_i} J(\theta) - \nabla_{\theta_i} J(\theta) | \leq \frac{2r_{max}}{(1-\gamma)^2}$ holds.
\end{lemma}

\begin{proof}
We begin the proof using the definition of the truncated and true gradient steps. The first inequality holds as a result of considering non-negative rewards, which means that the value functions and Q-functions are non-negative.

\begin{equation*}
\begin{split}
| \hat{\nabla}_{\theta_i} J(\theta_i) - \nabla_{\theta_i} J(\theta_i) | \stackrel{\text{def}}{=} &\Bigg| \frac{1}{1-\gamma} \Big( \hat{A}_i^{\pi}(s_{\mathcal{N}_i^{\kappa'}}, a_{\mathcal{N}_i^{\kappa'}}) - A_i^{\pi}(s, a) \Big) \Bigg| \\
= \frac{1}{1-\gamma} | \hat{Q}_i^{\pi}(s_{\mathcal{N}_i^{\kappa'}}, a_{\mathcal{N}_i^{\kappa'}}) &- \hat{V}_i^{\pi}(s_{\mathcal{N}_i^{\kappa'}}) - Q_i^{\pi}(s, a) + V_i^{\pi}(s) | \\
= \frac{1}{1-\gamma} | (\hat{Q}_i^{\pi}(s_{\mathcal{N}_i^{\kappa'}}, a_{\mathcal{N}_i^{\kappa'}})& - Q_i^{\pi}(s, a)) - (\hat{V}_i^{\pi}(s_{\mathcal{N}_i^{\kappa'}}) - V_i^{\pi}(s)) | \\
\stackrel{\text{triangle inequality}}{\leq} \frac{1}{1-\gamma} \Big( | \hat{Q}_i^{\pi}&(s_{\mathcal{N}_i^{\kappa'}}, a_{\mathcal{N}_i^{\kappa'}}) - Q_i^{\pi}(s, a)| + |\hat{V}_i^{\pi}(s_{\mathcal{N}_i^{\kappa'}}) - V_i^{\pi}(s) | \Big)\\
\stackrel{\text{from Lemma \ref{lemma:truncatedQ}}}{\leq} \frac{2k\gamma^{\kappa+1}}{1-\gamma} ~~~~~~~~ & 
\end{split}
\end{equation*}
\end{proof}

Now, we consider the proof of the main result in Theorem \ref{theorem:main} that follows from the intermediate lemmas.

\textbf{Proof of Main Result (Theorem \ref{theorem:main}):} The modified independent natural policy gradient algorithm in (\ref{eq:minpg}) converges to an $\epsilon$-equilibrium policy while using the approximated Q-function for $\kappa > 0$. The convergence of (\ref{eq:minpg}) using the parameterization in (\ref{eq:mpolicy}) is a result of the existence of a potential function in the game.

The modified independent natural policy gradient algorithm in (\ref{eq:minpg}) ensures that the potential function is non-decreasing at every time step, as described in Lemma \ref{lemma:incpot}. The proof follows the standard arguments of Lemma 3.4 in \cite{rf}.

\begin{lemma}
\label{lemma:incpot}
Assume that the agents use the modified independent policy gradient dynamics in (\ref{eq:minpg}), $\Phi^{\pi,t+1}(s) \geq \Phi^{\pi,t}(s)$ holds.
\end{lemma}

However, since the gradient step and the potential function are bounded, the potential function converges to a fixed point, as shown in Lemma \ref{lemma:fixed}.

\begin{lemma}
\label{lemma:fixed}
The modified independent policy gradient dynamics in (\ref{eq:minpg}) converges to fixed points.
\end{lemma}

\begin{proof}
By Lemma \ref{lemma:incpot}, we know that $\Phi^{\pi,t}(s)$ is non-decreasing with respect to time. Since the potential function is bounded, we conclude that $\Phi^{\pi,t}(s)$ converges as $t \rightarrow \infty$. Thus, the set of limit points of the policy $\pi$ is invariant, and $\Phi^{\pi,t}(s)$ is increasing except at the limit points. This means that the policy $\pi^t$ converges to an equilibrium policy.
\end{proof}

Thus, we showed that the potential function is strictly increasing except at the limit point, which is the convergence point corresponding to the equilibrium joint policy in Definition \ref{def:optpolicy}. However, we have to consider the fact that the utilization of the truncated advantage function in (\ref{eq:minpg}) means the fixed point (i.e. the convergence point) is in the $\epsilon$-neighborhood of the true fixed point. To characterize the value of $\epsilon$, we can bound the difference between the optimal true Q-function $Q^*$, calculated using the independent natural policy gradient algorithm~(\ref{eq:inpg2}), and the optimal truncated Q-function $\hat{Q}^*_i$, calculated using the modified independent natural policy gradient algorithm~(\ref{eq:minpg}), as shown in Lemma \ref{lemma:epsilon}.

\begin{lemma}
\label{lemma:epsilon}
The modified independent policy gradient dynamics in (\ref{eq:minpg}) converges to an $\epsilon$-neighborhood of the true fixed point.
\end{lemma}

\begin{proof}
Using Lemma \ref{lemma:truncatedQ}, we can bound the difference between the optimal truncated Q-function and the optimal true Q-function as follows:
\end{proof}

\begin{equation*}
\begin{split}
|\hat{Q}_i^{\pi*}(s_{\mathcal{N}_i^\kappa}, a_{\mathcal{N}_i^\kappa}) - Q_i^{\pi*}(s, a)| = \Big| \sum_{s'_{\mathcal{N}_{-i}^\kappa}, a'_{\mathcal{N}_{-i}^\kappa}} &w_i(s'_{\mathcal{N}_{-i}^\kappa}, a'_{\mathcal{N}_{-i}^\kappa}; s_{\mathcal{N}_{i}^\kappa}, a_{\mathcal{N}_{i}^\kappa}) Q_i^{\pi*}(s_{\mathcal{N}_{i}^\kappa}, s'_{\mathcal{N}_{-i}^\kappa}, a_{\mathcal{N}_{i}^\kappa}, a'_{\mathcal{N}_{-i}^\kappa}) - \\ 
&- Q_i^{\pi*}(s_{\mathcal{N}_{i}^\kappa}, s_{\mathcal{N}_{-i}^\kappa}, a_{\mathcal{N}_{i}^\kappa}, a_{\mathcal{N}_{-i}^\kappa}) \Big| \\
\stackrel{\text{from Lemma \ref{lemma:truncatedQ}}}{\leq} &k\gamma^{\kappa+1}
\end{split}
\end{equation*}

Finally, to get the value of $\epsilon$ as a function of $\kappa$, we set $\epsilon$ to the upper bound of the difference between the truncated and true optimal Q-functions since the upper bound defines the largest size of the neighborhood around the true optimal Q-function. This yields $\epsilon = \frac{r_{max}}{1-\gamma} \gamma^{\kappa+1}$. \rightline{$\blacksquare$}

\section{Illustrative Examples}
\label{sec:examples}
In this section, we present two illustrative examples to demonstrate the feasibility of the theory discussed in the previous sections. The examples aim to provide scenarios where the presented control strategies can be effectively implemented.

\subsection{Job Balancing Game Problem}
The first illustrative example considers a form of a job balancing game problem, in a networked structure. Fig. \ref{fig:ex1} displays the diagram of the utilized network, including the connections between the different agents. 

\begin{figure}
\begin{center}
\includegraphics[width=0.4\linewidth,keepaspectratio]{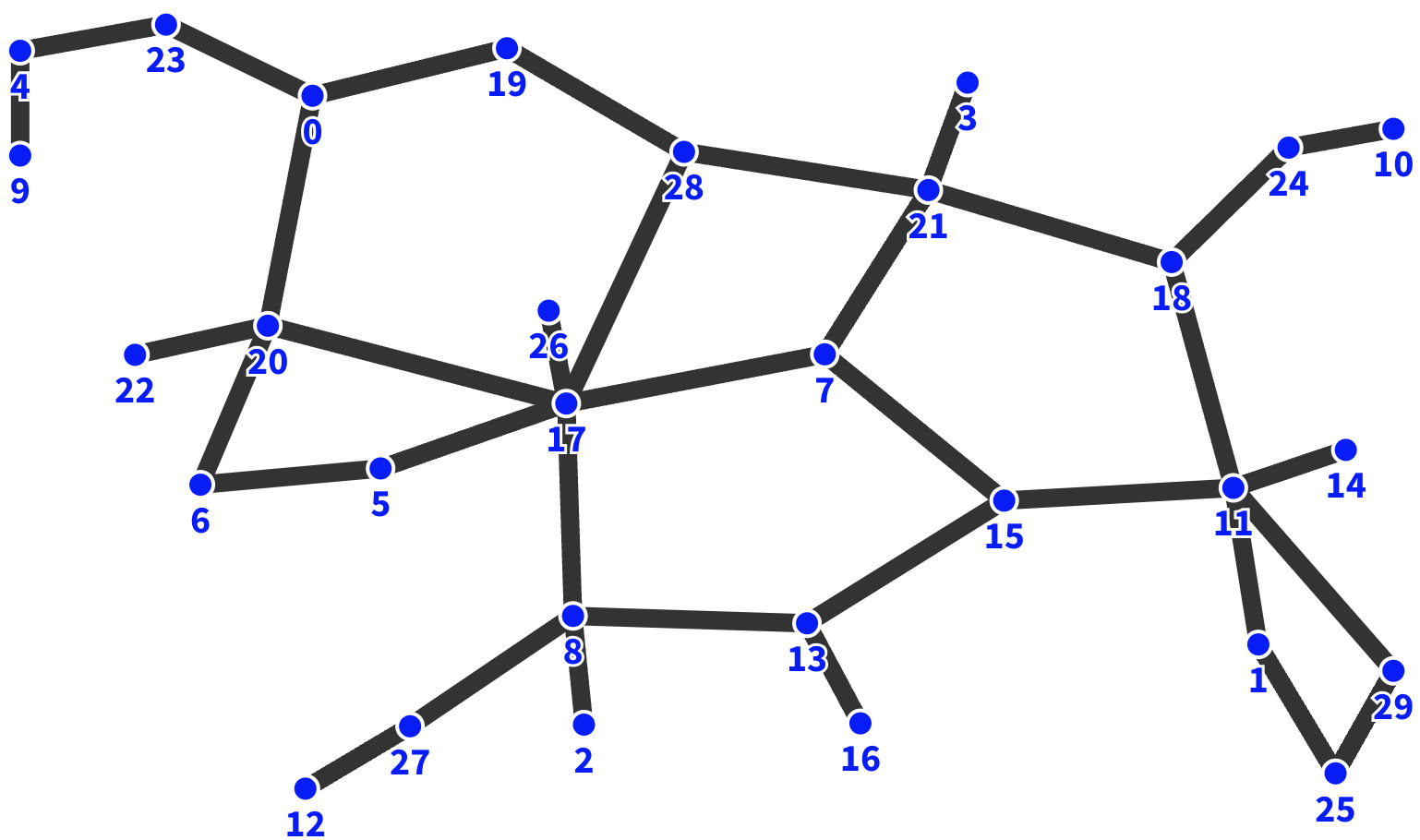}
\caption{Job Balancing Game Network Diagram}
\label{fig:ex1}
\end{center}
\end{figure}

For the job balancing game problem, we consider 30 agents (i.e. 30 nodes) that are expected to balance the total number of jobs in the network amongst all agents, which is assumed to be 60 jobs. This means that the congestion of jobs at a single node should be minimized. In this example, we assume that we have the network structure in Fig. \ref{fig:ex1}.

\textit{States and Dynamics}. The state of agent $i$ is the number of jobs currently being served by that agent. The initial state distribution is uniform. The state transitions of agent $i$ depends on the state and action of agent $i$, as well as the states and actions of its neighborhood, as shown in (\ref{eq:statetransex1}). For simplicity, the state transitions are assumed to be deterministic, meaning that each agent will transition from one specific state to another state with probability 1.

\begin{equation}
\label{eq:statetransex1}
s_i^{t+1} = s_i^t - a_i^t + \sum_{a \in a_{\mathcal{N}_i}} a_j
\end{equation}

\textit{Actions}. The action space of each agent is the number of jobs that it will delegate to neighboring agents, which is the same for all agents.

\textit{Rewards}. The reward function of each agent is defined in (\ref{eq:ex1_rew1}), where $s_i$ denotes the number of jobs assigned for agent $i$. This reward function is based on the deviation of the number of jobs for agent $i$ from the average number of jobs across all neighboring agents. The goal is to minimize this deviation to balance the workload. The discount factor is set to 0.9.

\begin{equation}
\label{eq:ex1_rew1}
r_i(s_{\mathcal{N}_i}, a_{\mathcal{N}_i}) =
    \begin{cases}
        \frac{1}{\Big| s_i - \frac{1}{|\mathcal{N}_i|} \sum_{j \in \mathcal{N}_i} s_j \Big|} & \text{if } s_i \neq \frac{1}{|\mathcal{N}_i^|} \sum_{j \in \mathcal{N}_i} s_j \\
        1 & \text{if } s_i = \frac{1}{|\mathcal{N}_i|} \sum_{j \in \mathcal{N}_i} s_j
    \end{cases}
\end{equation}

Fig. \ref{fig:ex1_res1} shows the convergence of the modified independent natural policy gradient algorithm for the job balancing game example. In this setting, when $\kappa$ is equal to 10, the agents have full state and action information. As expected, as the value of $\kappa$ increases, the amount of information available for the agent is higher, which decreases the error between the true optimal policy and the truncated optimal policy. Furthermore, the results prove the possibility of using truncated information in the setting of MPGs due to the relatively small error compared to the decreased information requirements.

\begin{figure}
\begin{center}
\includegraphics[width=0.4\linewidth,keepaspectratio]{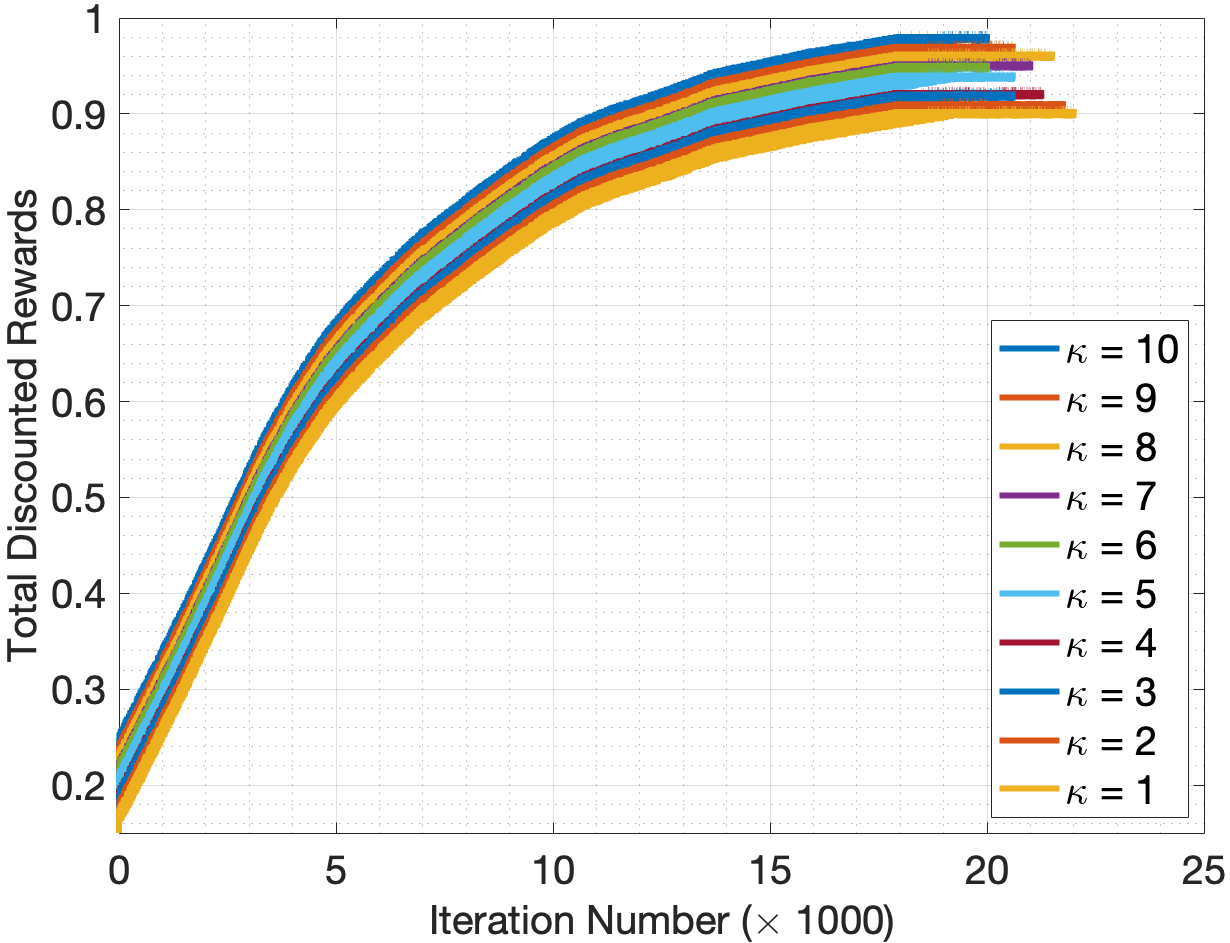}
\caption{Convergence Results of Independent Natural Policy Gradient for Job Balancing Game Problem}
\label{fig:ex1_res1}
\end{center}
\end{figure}

Fig. \ref{fig:ex1_res2} shows the percentage relative error based on $\kappa$ after the convergence of the modified independent natural policy gradient algorithm for the job balancing game example. As shown in the figure, as the value of $\kappa$ increases, the value of $\epsilon$ decreases as expected due to the availability of more information.

\begin{figure}
\begin{center}
\includegraphics[width=0.4\linewidth,keepaspectratio]{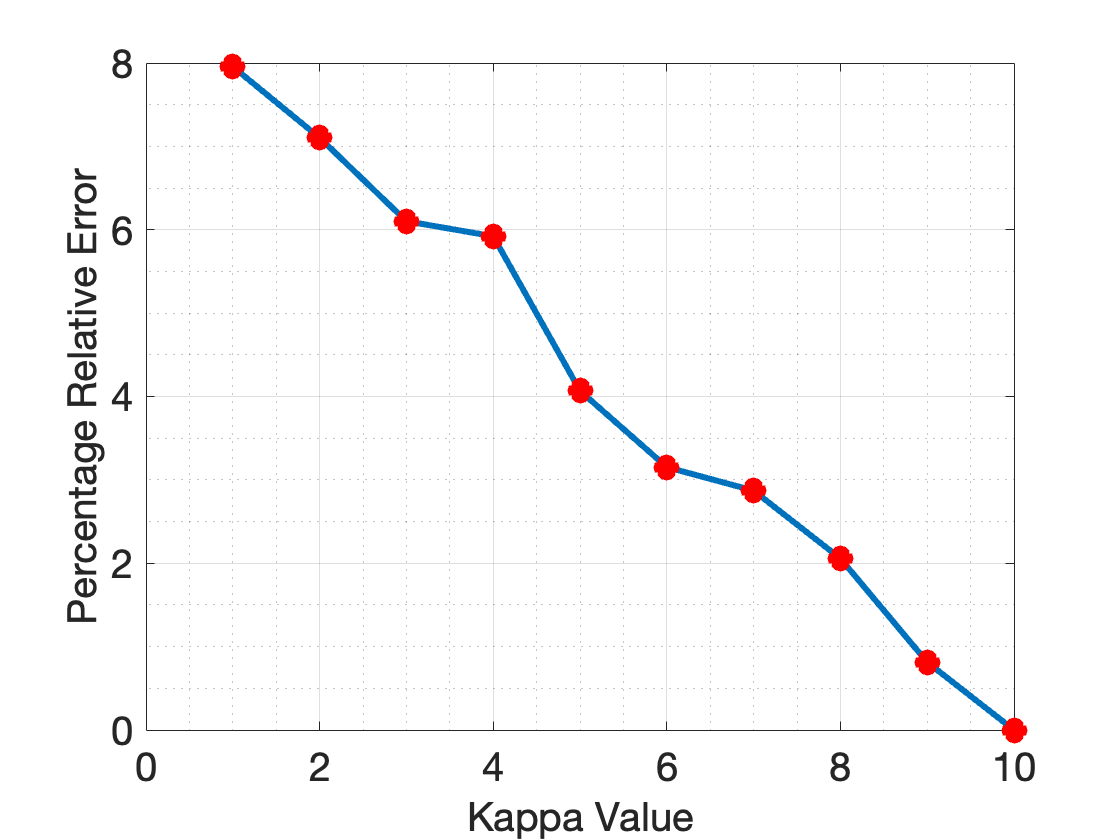}
\caption{Percentage relative error of $\epsilon$ based on $\kappa$ for Job Balancing Game Problem}
\label{fig:ex1_res2}
\end{center}
\end{figure}

\subsection{Sensor Coverage Problem}
The second illustrative example considers a sensor coverage problem. The formulation of the problem, including the agents' utility functions is inspired by the work in \cite{jm2, jm1}. Fig. \ref{fig:ex2} displays the diagram for the node network environment. As seen in the figure, for simplicity, we consider a ring environment with a total of 20 agents, and each agent is linked to two other agents only.

\textit{States and Dynamics}. The state of agent $i$ is the location of the agent, which is a finite set from 0 to 24, imitating the cells of a 5x5 grid world. The initial state distribution is uniform. The state transitions of agent $i$ occurs depending on the state and action of agent $i$ to move around the grid world. The state transitions are assumed to be probabilistic. Agents can only move to adjacent cells. If agents are not at the boundary cells, each action will take the agent to the expected cell with probability 0.85 and to one of the remaining three cells with probability 0.05 each. If agents are at the boundary, they will remain at their current state if they take an action that would have taken them outside the grid world.

\begin{figure}[!htbp]
\begin{center}
\includegraphics[width=0.4\linewidth,keepaspectratio]{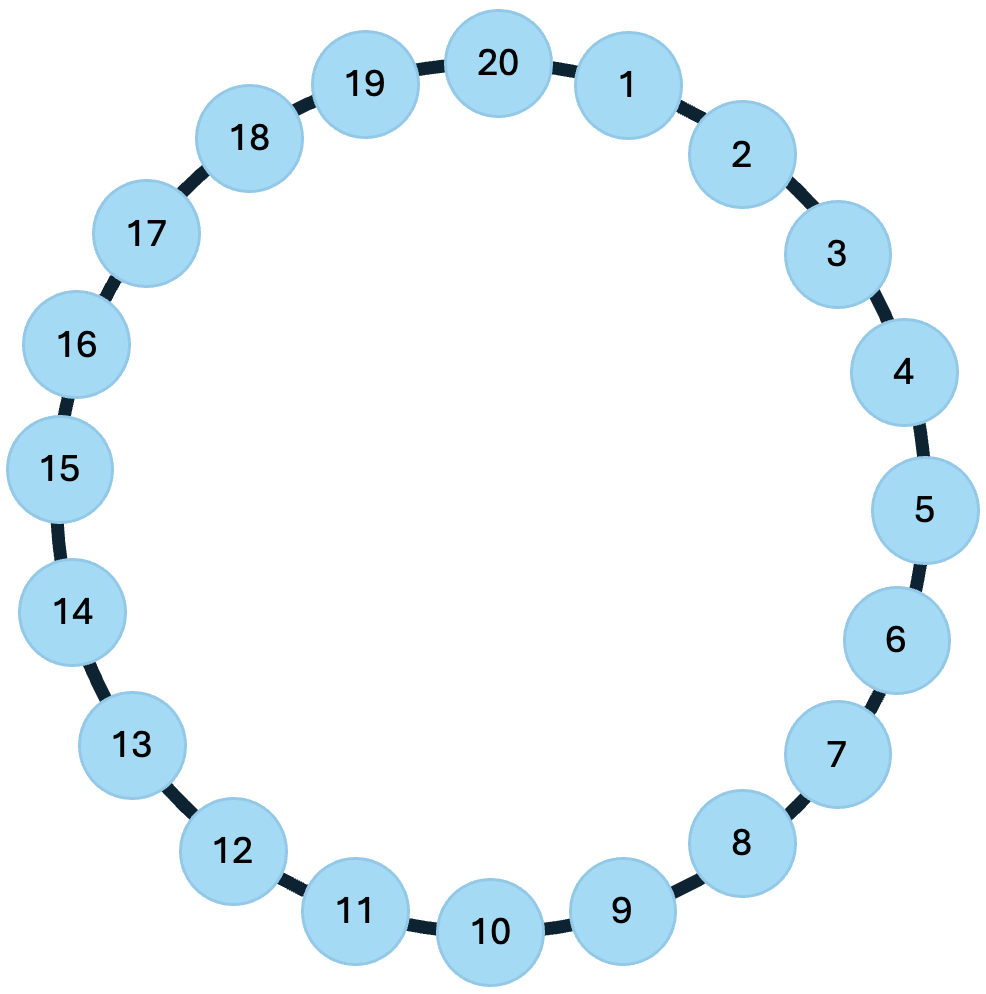}
\caption{Sensor Coverage Node Diagram}
\label{fig:ex2}
\end{center}
\end{figure}

\textit{Actions}. Further, the action space of of all agents is the same. More specifically, the actions of each agents corresponds to the movement across the environment, which can be right, left, up, and down.

\textit{Rewards}. The reward function of each agent is defined as shown in (\ref{eq:ex2_rew}), which is based on the work in \cite{jm2, jm1}. The discount factor is set to 0.9.

\begin{equation}
\label{eq:ex2_rew}
\begin{split}
r_i(s_{\mathcal{N}_i}, a_{\mathcal{N}_i}) &= \sum_{j \in \mathcal{N}_i} f_j(s_j, a_j) \\
&= \sum_{j \in s_{\mathcal{N}_i}} p_i(a_i, s_i) \prod_{j \in \mathcal{N}_i} (1-p_j(a_j, s_j)) \\
&= p_i(a_i, s_i) |\mathcal{N}_i| \prod_{j \in \mathcal{N}_i} (1-p_j(a_j, s_j))
\end{split}
\end{equation}

Fig. \ref{fig:ex2_res1} shows the convergence of the modified independent natural policy gradient algorithm for the sensor coverage example. In this setting, when $\kappa$ is equal to 10, the agents have full state and action information. As the value of $\kappa$ increases, the error between the true optimal policy and the truncated optimal policy decreases as expected. Additionally, the results prove the possibility of using truncated information in the setting of MPGs due to the relatively small error compared to the decreased information requirements.

\begin{figure}[!htbp]
\begin{center}
\includegraphics[width=0.4\linewidth,keepaspectratio]{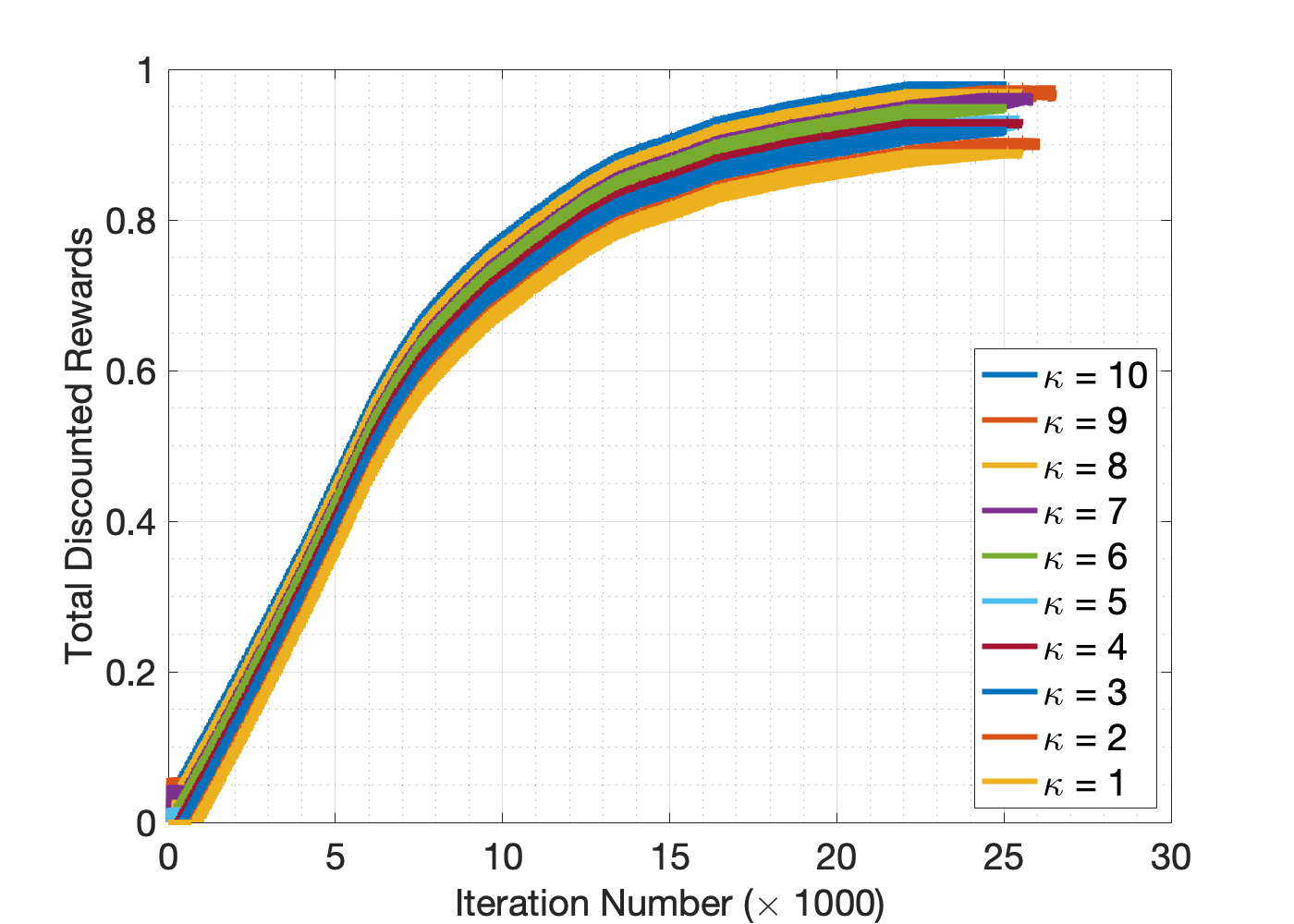}
\caption{Convergence Results of Independent Natural Policy Gradient for Sensor Coverage Problem}
\label{fig:ex2_res1}
\end{center}
\end{figure}

Fig. \ref{fig:ex2_res2} shows the percentage relative error based on $\kappa$ after the convergence of the modified independent natural policy gradient algorithm for the sensor coverage example. As seen in Fig. \ref{fig:ex2_res2}, as the value of $\kappa$ increases, the value of $\epsilon$ decreases due to the availability of more information.

\begin{figure}
\begin{center}
\includegraphics[width=0.4\linewidth,keepaspectratio]{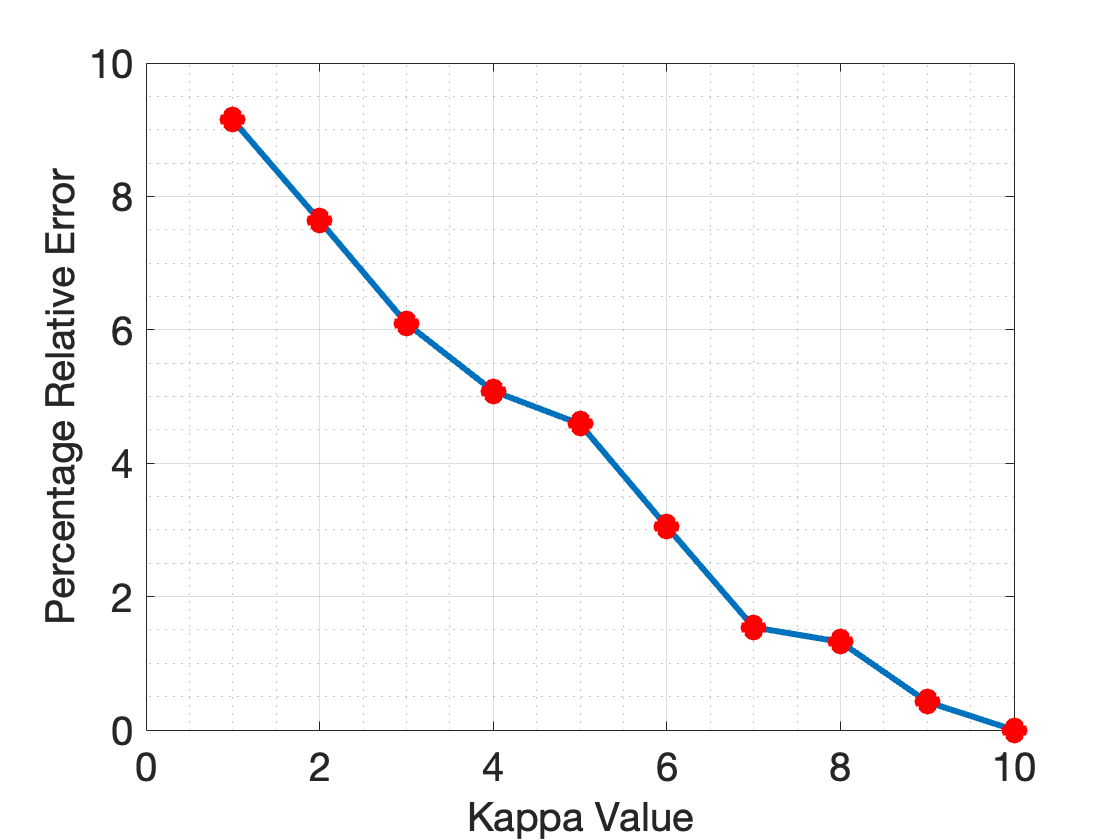}
\caption{Percentage relative error based on $\kappa$ for Sensor Coverage Problem}
\label{fig:ex2_res2}
\end{center}
\end{figure}

\section{Conclusion \& Future Work}
\label{sec:conclusions}
In conclusion, we demonstrated the significant prospects of distributed game theoretic approaches, particularly considering Markov potential games, in addressing the challenges of scalability and adaptability in multi-agent systems. By designing local agent objective functions that align with the global objectives, we ensure that distributed independent learning leads to convergence to optimal strategies for the entire system. The theoretical foundations and the illustrative example underline the effectiveness of our approach. This opens up promising avenues for further exploration and application across various domains where coordination of dynamic, multi-agent systems is crucial.

Future work can encompass the generalization of the problem to general Markov games, as well as the generalization of the model to continuous state and action spaces such as the LQ setting. Moreover, another research direction is the application of our methodology to more complex environments, thereby enhancing the robustness and applicability of distributed control laws in real-world scenarios.

\section*{Acknowledgments}
Start acknowledgments here...

\bibliographystyle{IEEEtran}
\bibliography{refs}

\end{document}